\documentclass[11pt]{article}
\usepackage{amssymb}
\makeatletter\@addtoreset {equation}{section}\makeatother

\newtheorem{theorem}{Theorem}
\newtheorem{lemma}{Lemma}
\newtheorem{remark}{Remark}

\newtheorem{corollary}{Corollary}

\newtheorem{proposition}{Proposition}

\newenvironment{proof}{
    \noindent {\it Proof.}}{\hfill$\Box$
}

\usepackage[dvips]{epsfig}
\usepackage{graphicx}

\begin{document}

\title{\bf Bounds on the tight-binding approximation for
the Gross--Pitaevskii equation with a periodic potential}

\author{Dmitry Pelinovsky$^1$ and Guido Schneider$^2$ \\
{\small $^1$ Department of Mathematics, McMaster University,
Hamilton, Ontario, Canada, L8S 4K1} \\
{\small $^2$ Institut f\"{u}r Analysis, Dynamik und Modellierung,
Universit\"{a}t Stuttgart} \\ {\small Pfaffenwaldring 57, D-70569
Stuttgart, Germany } }

\date{\today}
\maketitle

\begin{abstract}
We justify the validity of the discrete nonlinear Schr\"{o}dinger
equation for the tight-binding approximation in the context of the
Gross--Pitaevskii equation with a periodic potential. Our
construction of the periodic potential and the associated Wannier
functions is based on the previous work \cite{PelSchn}, while our
analysis involving energy estimates and Gronwall's inequality
addresses time-dependent localized solutions on large but finite
time intervals.
\end{abstract}

\section{Introduction}

We consider the Gross--Pitaevskii (GP) equation with a periodic
potential in the form
\begin{equation}
i \phi_t = - \phi_{xx} + V(x) \phi + \sigma |\phi|^2 \phi,
\label{GP}
\end{equation}
where the solution $\phi : \mathbb{R} \times \mathbb{R}_+ \mapsto
\mathbb{C}$ decays to zero sufficiently fast as $|x| \to \infty$,
the potential $V : \mathbb{R} \mapsto \mathbb{R}$ is a bounded
$2\pi$-periodic function, and the parameter $\sigma = \pm 1$ is
normalized for the cubic nonlinear term. In particular, we
consider the piecewise-constant potential in the form
\begin{equation}
\label{potential-V} V(x) = \left\{ \begin{array}{cc}
\varepsilon^{-2}, \quad & x \in (0,a) \; {\rm mod}(2\pi) \\ 0, \quad
& x \in (a,2\pi) \; {\rm mod}(2\pi)
\end{array} \right.
\end{equation}
for some fixed $0 < a < 2\pi$ and small $\varepsilon > 0$. The
asymptotic limit of small $\varepsilon$ represents the so-called
tight-binding approximation, for which the potential $V(x)$ is a
periodic sequence of large walls of a non-zero width and the
lowest bands in the spectrum of the linear operator $L = -
\partial_x^2 + V(x)$ are exponentially narrow with respect to
$\varepsilon$. According to the tight-binding approximation
\cite{AKKS}, time-dependent solutions of the GP equation
(\ref{GP}) are approximated by the time-dependent solutions of the
discrete nonlinear Schr\"{o}dinger (DNLS) equation in the form
\begin{equation}
\label{DNLS} i \dot{\phi}_n = \alpha \left( \phi_{n+1} + \phi_{n-1}
\right) + \sigma \beta |\phi_n|^2 \phi_n,
\end{equation}
where $\alpha$ and $\beta$ are $\varepsilon$-independent constants
and the sequence $\{ \phi_n(t) \}_{n \in \mathbb{Z}}$ represents a
small-amplitude solution $\phi(x,t)$ evaluated at the periodic
sequence of potential wells.

We proved in the previous work \cite{PelSchn} that stationary
localized solutions of the GP equation in the form $\phi(x,t) =
\Phi(x) e^{-i \omega t}$ with $\Phi \in H^1(\mathbb{R})$ and
$\omega \notin \sigma(L)$ are approximated for small values of
$\varepsilon$ by stationary localized solutions of the DNLS
equation in the form $\phi_n(t) = \Phi_n e^{-i \Omega t}$, where
$\vec{\mbox{\boldmath $\Phi$}} \in l^1(\mathbb{Z})$ and $\Omega$
is related to the rescaled parameter $\omega$. Here and
henceworth, we use the standard notations for the Sobolev space
$H^1(\mathbb{R})$ of scalar complex-valued functions equipped with
the squared norm
$$
\| \phi \|_{H^1(\mathbb{R})}^2 = \int_{\mathbb{R}} \left( |
\phi'(x)|^2 + |\phi(x)|^2 \right) dx
$$
and the space $l^1(\mathbb{Z})$ of vectors representing
complex-valued sequences equipped with the norm $\|
\vec{\mbox{\boldmath $\phi$}} \|_{l^1(\mathbb{Z})} = \sum_{n \in
\mathbb{Z}} |\phi_n|$. In this work, we extend our analysis to
time-dependent localized solutions of these equations and prove
that the formal tight-binding approximation of \cite{AKKS} can be
justified for small values of $\varepsilon$ on large but finite
time intervals. It will be clear from our analysis that the
appropriate space for the time-dependent localized solution of the
GP equation (\ref{GP}) is associated with the quadratic form
generated by operator $-\partial_x^2 + V(x) + 1$. We denote this
space by ${\cal H}^1(\mathbb{R})$ and equip it with the squared
norm
\begin{equation}
\label{norm-space} \| \phi \|_{{\cal H}^1(\mathbb{R})}^2 =
\int_{\mathbb{R}} \left( | \phi'(x)|^2 + V(x) |\phi(x)|^2 +
|\phi(x)|^2 \right) dx.
\end{equation}
Since $V(x) \geq 0$ for all $x \in \mathbb{R}$, it is clear that
$\| \phi \|_{H^1(\mathbb{R})} \leq \| \phi \|_{{\cal
H}^1(\mathbb{R})}$.

Our analysis is closely related to the recent works on
justifications of nonlinear evolution equations for pulses that
exist in space-periodic potentials near edges of spectral bands
\cite{BSTU} and in narrow band gaps of one-dimensional \cite{SU} and
two-dimensional \cite{PSD} potentials. A similar work in the context
of a nonlinear heat equation with a periodic diffusive term was
developed in \cite{SV} with the invariant manifold reductions.
Although the justification of lattice equations for the
time-dependent solutions of dissipative (reaction--diffusion)
equations can be extended globally for $t \geq 0$, the justification
of the DNLS equation can only be carried out for finite time
intervals because the GP equation is a conservative (Hamiltonian)
system. Reductions to the DNLS equation on a finite lattice for a
finite time interval were also discussed in the context of the GP
equation with a $N$-well trapping potential \cite{BS}.

Methods of our analysis follow closely to arguments from
\cite{KSM} and rely on the Wannier function decomposition from
\cite{PelSchn} as well as on energy estimates and Gronwall's
inequality. The Wannier function decomposition is reviewed in
Section 2. The energy estimates and the bounds on the remainder
terms are studied in Section 3. The main theorem is formulated in
Section 2 and proved in Section 3.

\section{Wannier function decomposition}

Let $u_l(x;k)$ be a Bloch function of the operator $L = -
\partial_x^2 + V(x)$ for the eigenvalue $\omega_l(k)$, such that $l
\in \mathbb{N}$, $k \in \mathbb{T} = \left[
-\frac{1}{2},\frac{1}{2}\right] {\rm mod}(1)$, $u_l(x+2\pi;k) =
u_l(x;k) e^{i 2\pi k}$ for all $x \in \mathbb{R}$, and the following
orthogonality and normalization conditions are met
\begin{equation}
\label{orthogonality} \int_{\mathbb{R}} \bar{u}_{l'}(x,k')u_l(x,k)
dx = \delta_{l,l'} \delta(k - k'), \quad \forall l,l' \in
\mathbb{N}, \;\; \forall k,k' \in \mathbb{T},
\end{equation}
where $\delta_{l,l'}$ is the Kronecker symbol and $\delta(k)$ is
the Dirac delta function in the sense of distributions. To
normalize uniquely the phase factors of the Bloch functions
\cite{panati}, we assume that $u_l(x;-k) = \bar{u}_l(x;k)$ is
chosen as a Bloch function for $\omega_l(-k) = \bar{\omega}_l(k) =
\omega_l(k)$.

Since the band function $\omega_l(k)$ and the Bloch function
$u_l(x;k)$ are periodic with respect to $k \in \mathbb{T}$ for any
$l \in \mathbb{N}$, we represent them by the Fourier series
\begin{equation}
\label{band-function} \omega_l(k) = \sum_{n \in \mathbb{Z}}
\hat{\omega}_{l,n} e^{i 2\pi n k}, \quad u_l(x;k) = \sum_{n \in
\mathbb{Z}} \hat{u}_{l,n}(x) e^{i 2\pi n k}, \quad \forall l \in
\mathbb{N}, \; \forall k \in \mathbb{T},
\end{equation}
where the coefficients satisfy the constraints
\begin{equation}
\hat{\omega}_{l,n} = \hat{\bar{\omega}}_{l,-n} =
\hat{\omega}_{l,-n}, \quad \hat{u}_{l,n}(x) =
\hat{\bar{u}}_{l,n}(x), \quad \forall n \in \mathbb{Z}, \;\; \forall
l \in \mathbb{N}, \;\; \forall x \in \mathbb{R}
\end{equation}
and
\begin{equation}
\label{Wannier-shifts} \hat{u}_{l,n}(x) = \hat{u}_{l,n-1}(x-2\pi) =
\hat{u}_{l,0}(x-2 \pi n), \qquad \forall n \in \mathbb{Z}, \;
\forall l \in \mathbb{N}, \; \forall x \in \mathbb{R}.
\end{equation}
The real-valued functions $\hat{u}_{l,n}(x)$ are referred to as
the Wannier functions. The following two propositions from
\cite{PelSchn} summarize properties of the band and Wannier
functions for the potential $V(x)$ given by (\ref{potential-V}) in
the limit of small $\varepsilon > 0$.

\begin{proposition}
\label{proposition-bands} Let $V$ be given by (\ref{potential-V}).
For any fixed $l \in \mathbb{N}$, there exist $\varepsilon_0 > 0$
and $\varepsilon$-independent constants $\zeta_0, \omega_0, c_n >
0$, such that, for any $\varepsilon \in [0,\varepsilon_0)$, the band
functions of the operator $L = -\partial_x^2 + V(x)$ satisfy the
properties:
\begin{eqnarray}
\label{potential-1} & (i) & \mbox{(band separation)} \quad
\phantom{0 \leq \epsilon < \epsilon_0 : \quad} \min_{\forall m \in
\mathbb{N} \backslash \{l\}} \inf_{\forall k \in \mathbb{T}}
|\omega_m(k) - \hat{\omega}_{l,0} | \geq \zeta_0,
\\ & (ii) & \mbox{(band boundness)} \phantom{\forall 0 \leq
\epsilon < \epsilon_0 : \quad text} |\hat{\omega}_{l,0}| \leq \omega_0,  \label{potential-3a} \\
& (iii) & \mbox{(tight-binding approximation)} \quad
\phantom{text} | \hat{\omega}_{l,n} | \leq c_n \varepsilon^n
e^{-\frac{n a}{\varepsilon}}, \quad n \in \mathbb{N}.
\label{potential-3}
\end{eqnarray}
\end{proposition}

\begin{proposition}
\label{proposition-Wannier-limit} Let $V$ be given by
(\ref{potential-V}). For any fixed $l \in \mathbb{N}$, there exists
$\varepsilon_0 > 0$ and $\varepsilon$-independent constants $U_0,
C_0, C_n > 0$, such that, for any $\varepsilon \in
[0,\varepsilon_0)$, the Wannier functions of the operator $L =
-\partial_x^2 + V(x)$ satisfy the properties:
\begin{eqnarray}
\label{L-4-bound} & (i) & \mbox{(boundness of norms)} \quad
\phantom{ \forall 0 }
\| \hat{u}_{l,n} \|_{{\cal H}^1(\mathbb{R})} \leq U_0, \quad n \in \mathbb{N}, \\
& (ii) & \mbox{(compact support)} \quad \phantom{textte} |
\hat{u}_{l,0}(x) - \hat{u}_0(x) | \leq C_0 \varepsilon,
\;\;\forall x \in [0,2\pi],
\label{compact-Wannier} \\
& (iii) & \mbox{(exponential decay)} \quad \phantom{text} \left|
\hat{u}_{l,0}(x) \right| \leq C_n \varepsilon^n e^{-\frac{n
a}{\varepsilon}},   \label{decay-Wannier-strong} \\  \nonumber &
\phantom{t} & \phantom{texttexttexttexttexttexttexttext} \forall x
\in [-2\pi n, -2\pi(n-1)] \cup [2 \pi n,2\pi (n+1)], \;\; n \in
\mathbb{N},
\end{eqnarray}
where
\begin{equation}
\label{asymptotic-Wannier-function} \hat{u}_0(x) = \left\{
\begin{array}{ll} 0, \;\; & \forall
x \in [0,a], \\
\frac{\sqrt{2}}{\sqrt{2\pi - a}} \sin \frac{\pi l (2 \pi - x)}{2 \pi
- a}, \;\; & \forall x \in [a,2\pi].
\end{array} \right.
\end{equation}
\end{proposition}

Figure \ref{fig-roots} illustrates the spectrum of $L$ and the
Wannier functions for the potential $V$ in (\ref{potential-V})
with $a = \pi$ and $\varepsilon = 0.5$. The left panel shows the
first spectral bands of $L$ computed from the trace of the
monodromy matrix \cite{PelSchn}. The right panel shows the Wannier
function $\hat{u}_{1,0}(x)$ computed by using the integral
representation $\hat{u}_{1,0}(x) = \int_{\mathbb{T}} u_1(x;k) dk$
and the finite-difference approximation of the Bloch function
$u_1(x;k)$. The solid lines for the Wannier function
$\hat{u}_{1,0}(x)$ approach the dotted line for the asymptotic
approximation (\ref{asymptotic-Wannier-function}) with $l = 1$ as
$\varepsilon$ gets smaller.

\begin{figure}[htbp]
\begin{center}
\includegraphics[height=6cm]{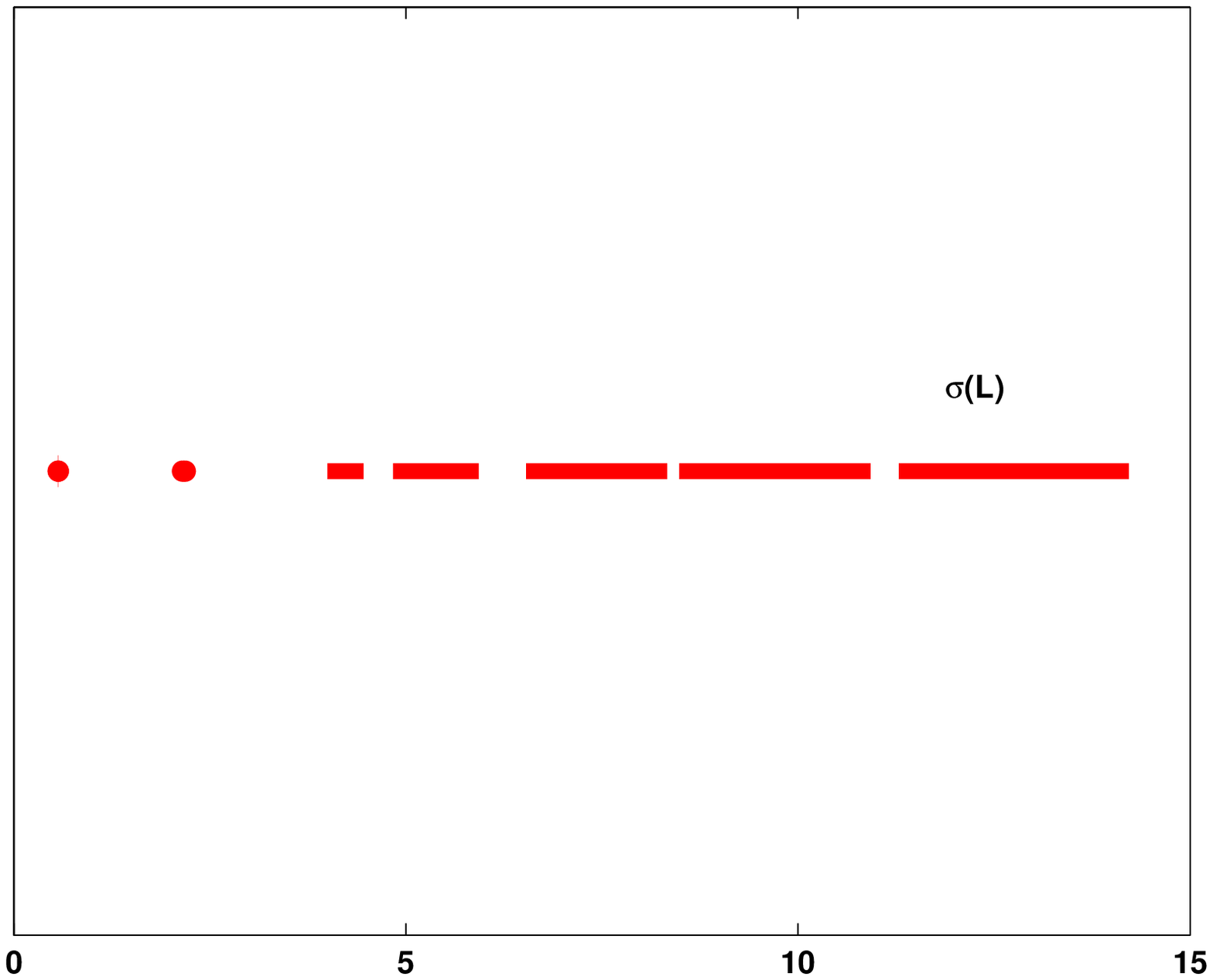}
\includegraphics[height=6cm]{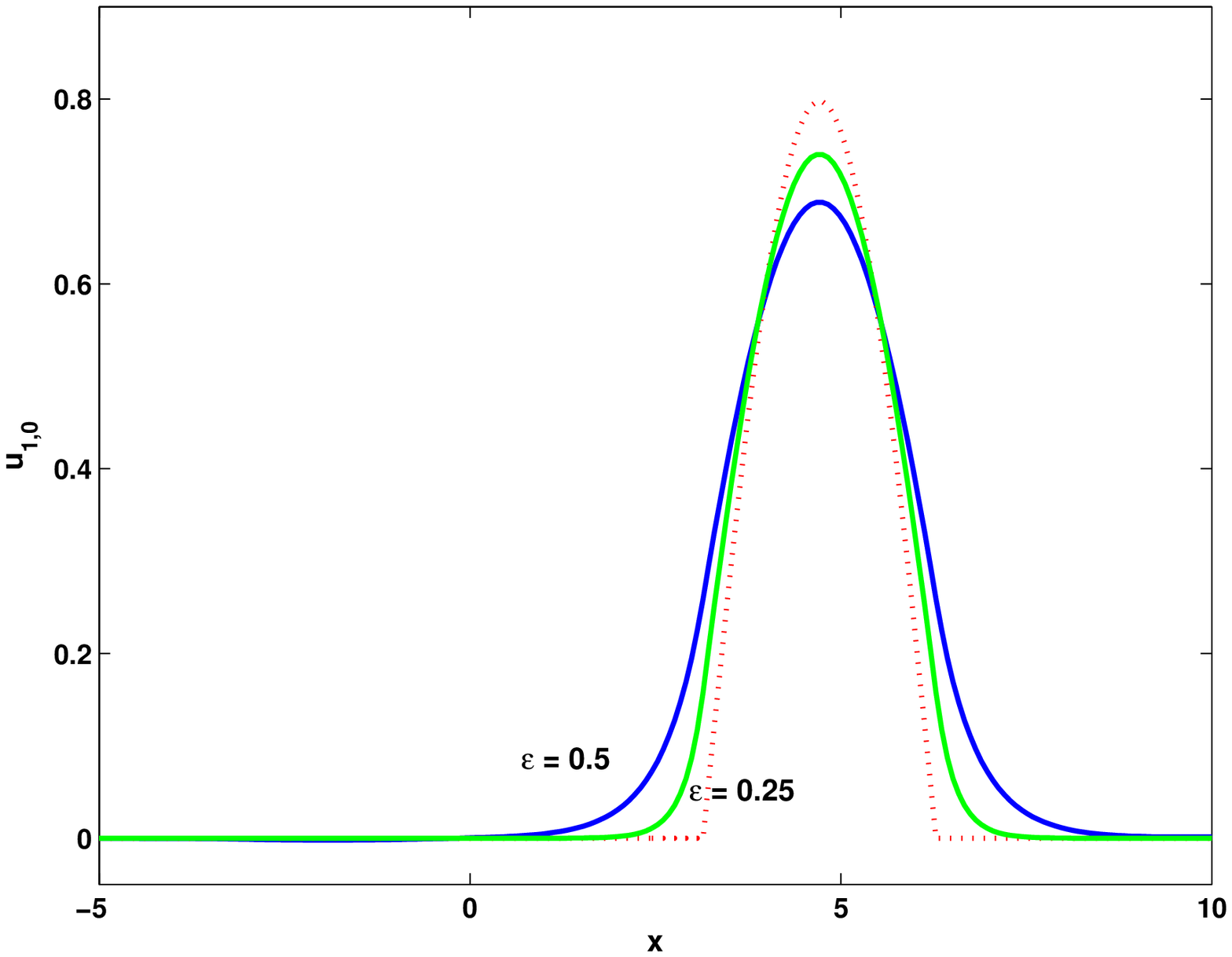}
\end{center}
\caption{Left: the band-gap structure of the spectrum $\sigma(L)$
for $\varepsilon = 0.5$. Right: the Wannier functions
$\hat{u}_{1,0}(x)$ for $\varepsilon = 0.5$ and $\varepsilon = 0.25$.
The dotted line shows the asymptotic approximation
(\ref{asymptotic-Wannier-function}).} \label{fig-roots}
\end{figure}

We use the set of Wannier functions $\{ \hat{u}_{l,n} \}_{n \in
\mathbb{Z}}$ for a fixed $l \in \mathbb{N}$ to represent formally
solutions of the GP equation (\ref{GP}) in the form
\begin{equation}
\phi(x,t) = \mu^{1/2} \left( \varphi_0(x,T) + \mu \varphi(x,t)
\right) e^{- i \hat{\omega}_{l,0} t}, \quad T = \mu t, \quad  \mu
= \varepsilon e^{-\frac{a}{\varepsilon}},
\end{equation}
where
\begin{equation}
\varphi_0(x,T) = \sum_{n \in \mathbb{Z}} \phi_n(T)
\hat{u}_{l,n}(x).
\end{equation}
Let us assume that the sequence $\{ \phi_n \}_{n \in \mathbb{Z}}$
satisfies the DNLS equation (\ref{DNLS}) with $\alpha =
\frac{\hat{\omega}_{l,1}}{\mu}$ and $\beta = \| \hat{u}_{l,0}
\|^4_{L^4(\mathbb{R})}$. Since the Wannier functions satisfy the
ODE system from \cite{PelSchn}
\begin{equation}
- \hat{u}_{l,n}''(x) + V(x) \hat{u}_{l,n}(x) = \sum_{n' \in
\mathbb{Z}} \hat{\omega}_{l,n-n'} \hat{u}_{l,n'}(x), \qquad \forall
n \in \mathbb{Z}, \label{system-Wannier-Fourier}
\end{equation}
we obtain an inhomogeneous PDE system for the function
$\varphi(x,t)$ in the form
\begin{eqnarray}
\nonumber i \varphi_t & = & - \varphi_{xx} + V(x) \varphi -
\hat{\omega}_{l,0} \varphi + \frac{1}{\mu} \sum_{n \in \mathbb{Z}}
\sum_{m \geq 2} \hat{\omega}_{l,m} \left( \phi_{n+m} + \phi_{n-m}
\right) \hat{u}_{l,n} \\ & \phantom{t} &
\phantom{texttexttexttexttexttext} + \sigma \left( | \varphi_0 + \mu
\varphi|^2 \left( \varphi_0 + \mu \varphi \right) - \beta \sum_{n
\in \mathbb{Z}} |\phi_n |^2 \phi_n \hat{u}_{l,n} \right).
\label{evolution-time}
\end{eqnarray}
The term $|\varphi_0|^2 \varphi_0$ gives projections both to the
selected $l$-th spectral band and to its complement in
$L^2(\mathbb{R})$. The following two lemmas allow us to control
both projections.

\begin{lemma}
\label{proposition-Wannier} Let $E_l$ be the invariant closed
subspace of $L^2(\mathbb{R})$ associated with the $l$-th spectral
band and assume that $E_l \cap E_m = \emptyset$ for a fixed $l \in
\mathbb{N}$ and all $m \neq l$. Then, $\langle \hat{u}_{n,l},
\hat{u}_{n',l} \rangle = \delta_{n,n'}$ for any $n,n' \in
\mathbb{Z}$ and there exists constants $\eta_l > 0$ and $C_l
> 0$, such that
\begin{equation}
\label{decay-Wannier} |\hat{u}_{l,n}(x)| \leq C_l e^{-\eta_l |x - 2
\pi n |}, \quad \forall n \in \mathbb{Z}, \;\; \forall x \in
\mathbb{R}.
\end{equation}
Moreover, if $\vec{\mbox{\boldmath $\phi$}} \in l^1(\mathbb{Z})$,
$\hat{u}_{l,n} \in {\cal H}^1(\mathbb{R})$, and $\phi(x) = \sum_{n
\in \mathbb{Z}} \phi_n \hat{u}_{l,n}(x)$ for a fixed $l \in
\mathbb{N}$, then $\phi \in E_l$, $(\phi, \psi) = 0$, $\forall \psi
\in  \cup_{m \neq l} E_m$, and $\phi \in {\cal H}^1(\mathbb{R})$.
\end{lemma}

\begin{proof}
The orthogonality and exponential decay of Wannier functions
follows from the orthogonality relations (\ref{orthogonality}) and
complex integration (see \cite{PelSchn} for the proof). The
assertion that $\phi \in E_l$ and $(\phi, \psi) = 0$, $\forall
\psi \in  \cup_{m \neq l} E_m$ follows from the $L^2$ spectral
theory for the  operator $L = -\partial_x^2 + V(x)$ (if
$\vec{\mbox{\boldmath $\phi$}} \in l^1(\mathbb{Z})$, then
$\vec{\mbox{\boldmath $\phi$}} \in l^2(\mathbb{Z})$ and $\phi \in
L^2(\mathbb{R})$). The assertion that $\phi \in {\cal
H}^1(\mathbb{R})$ follows from the triangular inequality.
\end{proof}

\begin{remark}
{\rm According to property (i) of Proposition
\ref{proposition-bands}, the $l$-th spectral band for a fixed $l
\in \mathbb{N}$ becomes disjoint from the rest of the spectrum of
$L$ in the asymptotic limit $\varepsilon \to 0$.  According to
property (i) of Proposition \ref{proposition-Wannier-limit},
$\hat{u}_{l,n} \in {\cal H}^1(\mathbb{R})$ for all $n \in
\mathbb{N}$ uniformly in $\varepsilon \geq 0$. Therefore, the
assumptions of Lemma \ref{proposition-Wannier} are satisfied for
sufficiently small $\varepsilon \geq 0$.}
\end{remark}

\begin{lemma}
\label{lemma-LS-reductions-estimate} Let $\Pi$ be an orthogonal
projection from $L^2(\mathbb{R})$ to $E_l \subset
L^2(\mathbb{R})$. There exists a unique solution $\varphi \in
{\cal H}^1(\mathbb{R})$ of the inhomogeneous equation
\begin{equation}
\label{inhomogeneous-equation} \left( -\partial^2_x + V(x) -
\hat{\omega}_{l,0} \right) \varphi = \left( {\cal I} - \Pi \right)
f,
\end{equation}
for any $f \in L^2(\mathbb{R})$, uniformly in $\varepsilon \geq
0$, such that $(\varphi,\psi) = 0$, $\forall \psi \in E_l$.
\end{lemma}

\begin{proof}
By property (i) of Proposition \ref{proposition-bands}, if $f \in
L^2(\mathbb{R})$, then $\varphi \in L^2(\mathbb{R})$ uniformly in
$\varepsilon \geq 0$. By property (ii) of Proposition
\ref{proposition-bands}, we obtain
\begin{equation}
\label{equation-Y} \| \varphi' \|^2_{L^2(\mathbb{R})} + \| V^{1/2}
\varphi \|^2_{L^2(\mathbb{R})} \leq |\hat{\omega}_{l,0}| \|
\varphi \|^2_{L^2(\mathbb{R})} + |(\varphi,f)| \leq C \| f
\|^2_{L^2(\mathbb{R})},
\end{equation}
where the constant $C > 0$ is $\varepsilon$-independent.
Therefore, if $f \in L^2(\mathbb{R})$, then $\varphi \in {\cal
H}^1(\mathbb{R})$ uniformly in $\varepsilon \geq 0$. Uniqueness of
$\varphi$ follows from the fact that the operator $L -
\hat{\omega}_{l,0}$ is invertible in $L^2(\mathbb{R}) \backslash
E_l$.
\end{proof}

We can also use the following elementary result.
\begin{lemma}
\label{lemma-Banach-algebra} The space ${\cal H}^1(\mathbb{R})$
forms Banach algebra under the pointwise multiplication, such that
\begin{equation}
\forall u,v \in {\cal H}^1(\mathbb{R}) : \qquad \| u v \|_{{\cal
H}^1(\mathbb{R})} \leq C \| u \|_{{\cal H}^1(\mathbb{R})} \| v
\|_{{\cal H}^1(\mathbb{R})},
\end{equation}
for some $C > 0$.
\end{lemma}

\begin{proof}
The result follows from the representation $\| u \|^2_{{\cal
H}^1(\mathbb{R})} = \| u \|^2_{H^1(\mathbb{R})} + \| V^{1/2} u
\|^2_{L^2(\mathbb{R})}$ and the Sobolev embedding theorem $\| u
\|_{L^{\infty}(\mathbb{R})} \leq C \| u \|_{H^1(\mathbb{R})}$ for
some $C > 0$.
\end{proof}

Let us return back to the evolution problem (\ref{evolution-time})
and decompose the solution $\varphi(x,t)$ into two parts
$\varphi(x,t) = \varphi_1(x,T) + \psi(x,t)$, where $\varphi_1$ is
a solution of the inhomogeneous equation
(\ref{inhomogeneous-equation}) with $f = - \sigma
|\varphi_0(x,T)|^2 \varphi_0(x,T)$, while $\psi$ satisfies the
evolution problem in the abstract form
\begin{eqnarray}
\label{evolution-psi} i \psi_t = \left( L - \hat{\omega}_{l,0}
\right) \psi + \mu R(\mbox{\boldmath $\vec{\phi}$}) + \mu \sigma
N(\mbox{\boldmath $\vec{\phi}$},\psi),
\end{eqnarray}
where
\begin{eqnarray}
R(\mbox{\boldmath $\vec{\phi}$}) = \frac{1}{\mu^2} \sum_{n \in
\mathbb{Z}} \sum_{m \geq 2} \hat{\omega}_{l,m} \left( \phi_{n+m} +
\phi_{n-m} \right) \hat{u}_{l,n} + \frac{\sigma}{\mu} \left( \Pi
|\varphi_0|^2 \varphi_0 - \beta \sum_{n \in \mathbb{Z}} |\phi_n |^2
\phi_n \hat{u}_{l,n} \right)
\end{eqnarray}
and
\begin{eqnarray}
\nonumber N(\mbox{\boldmath $\vec{\phi}$},\psi) & = & - i \sigma
\partial_T \varphi_1 + 2 |\varphi_0 |^2 (\varphi_1 + \psi) +
\varphi_0^2 (\bar{\varphi}_1 + \bar{\psi}) \\ & \phantom{t} & +
\mu \left( 2 |\varphi_1 + \psi|^2 \varphi_0 + (\varphi_1 + \psi)^2
\bar{\varphi}_0 \right) + \mu^2 |\varphi_1 + \psi|^2 (\varphi_1 +
\psi),
\end{eqnarray}
with $\varphi_0 = \sum_{n \in \mathbb{N}} \phi_n \hat{u}_{l,n}$,
$\varphi_1 = -\sigma ({\cal I} - \Pi) (L -
\hat{\omega}_{l,0})^{-1} ({\cal I} - \Pi) |\varphi_0|^2
\varphi_0$, and $\sigma = \pm 1$. The following lemma gives a
bound on the vector field of the evolution problem
(\ref{evolution-psi}).

\begin{lemma}
\label{lemma-reductions} Let $D_{\delta_1} \subset l^1(\mathbb{Z})$
be a ball of finite radius $\delta_1$ centered at $0 \in
l^1(\mathbb{Z})$, $D_{\delta_2} \subset {\cal H}^1(\mathbb{R})$ be a
ball of finite radius $\delta_2$ centered at $0 \in {\cal
H}^1(\mathbb{R})$ and $R_{\mu_0} \subset \mathbb{R}$ be an interval
of small radius $\mu_0$ centered at $0 \in \mathbb{R}$. Then, for
any $\mu \in (0,\mu_0)$, $\| \mbox{\boldmath $\vec{\phi}$}
\|_{l^1(\mathbb{Z})} \in [0,\delta_1)$ and $\| \psi \|_{{\cal
H}^1(\mathbb{R})} \in [0,\delta_2)$, there exists $\mu$-independent
constants $C_R,C_N > 0$ such that
\begin{equation}
\label{bound-map} \| R(\mbox{\boldmath $\vec{\phi}$}) \|_{{\cal
H}^1(\mathbb{R})} \leq C_R \| \vec{\mbox{\boldmath $\phi$}}
\|_{l^1(\mathbb{Z})}, \quad \| N(\vec{\mbox{\boldmath $\phi$}},\psi)
\|_{{\cal H}^1(\mathbb{R})} \leq C_N \left( \| \vec{\mbox{\boldmath
$\phi$}} \|_{l^1(\mathbb{Z})} + \| \psi \|_{{\cal H}^1(\mathbb{R})}
\right).
\end{equation}
\end{lemma}

\begin{proof}
By the last assertion of Lemma \ref{proposition-Wannier}, if
$\vec{\mbox{\boldmath $\phi$}} \in l^1(\mathbb{Z})$ and $\phi(x) =
\sum_{n \in \mathbb{Z}} \phi_n \hat{u}_{l,n}(x)$ for a fixed $l
\in \mathbb{N}$, then $\phi \in {\cal H}^1(\mathbb{R})$.
Therefore, there exists $C_0 > 0$ such that $\| \varphi_0
\|_{{\cal H}^1(\mathbb{R})} \leq C_0  \| \mbox{\boldmath
$\vec{\phi}$} \|_{l^1(\mathbb{Z})}$. By Lemmas
\ref{lemma-LS-reductions-estimate} and \ref{lemma-Banach-algebra},
there exists $C_1 > 0$ such that $\| \varphi_1 \|_{{\cal
H}^1(\mathbb{R})} \leq C_1  \| \mbox{\boldmath $\vec{\phi}$}
\|^3_{l^1(\mathbb{Z})}$, where we have used the fact that $\|
|\varphi_0|^2 \varphi_0 \|_{{\cal H}^1(\mathbb{R})} \leq C^3 \|
\varphi_0 \|^3_{{\cal H}^1(\mathbb{R})}$ for some $C > 0$. The
vector field $R(\mbox{\boldmath $\vec{\phi}$})$ can be represented
by $R(\mbox{\boldmath $\vec{\phi}$}) = \sum_{n \in \mathbb{Z}}
r_n(\mbox{\boldmath $\vec{\phi}$}) \hat{u}_{l,n}(x)$, where
$$
r_n(\mbox{\boldmath $\vec{\phi}$}) = \frac{1}{\mu^2} \sum_{m \geq 2}
\hat{\omega}_{l,m} \left( \phi_{n+m} + \phi_{n-m} \right) +
\frac{\sigma}{\mu} \sum_{(n_1,n_2,n_3) \in \mathbb{Z}^3 \backslash
\{(n,n,n)\}} K_{n,n_1,n_2,n_3} \phi_{n_1} \bar{\phi}_{n_2}
\phi_{n_3},
$$
where $K_{n,n_1,n_2,n_3} = (\hat{u}_{l,n}, \hat{u}_{l,n_1}
\hat{u}_{l,n_2} \hat{u}_{l,n_3})$. The first bound in
(\ref{bound-map}) is proved if $\vec{\bf r} \in l^1(\mathbb{Z})$
for every $\mbox{\boldmath $\vec{\phi}$} \in l^1(\mathbb{Z})$ and
the map $\vec{\bf r}(\mbox{\boldmath $\vec{\phi}$})$ is uniformly
bounded for small $\mu > 0$. The first term in $\vec{\bf
r}(\mbox{\boldmath $\vec{\phi}$})$ is estimated as follows
$$
\left\| \sum_{m \geq 2} \hat{\omega}_{l,m} \left( \phi_{n+m} +
\phi_{n-m} \right) \right\|_{l^1(\mathbb{Z})} \leq \sum_{n \in
\mathbb{Z}} \sum_{m \in \mathbb{Z} \backslash \{0,1,-1\}}
|\hat{\omega}_{l,m+n}| |\phi_n| \leq K_1 \| \vec{\mbox{\boldmath
$\phi$}} \|_{l^1(\mathbb{Z})},
$$
where $K_1 = {\rm sup}_{n \in \mathbb{Z}\backslash \{0,1,-1\}}
\sum_{m \in \mathbb{Z}} |\hat{\omega}_{l,n+m} |$. Since
$\omega_{l}(k)$ is analytically extended along the Riemann surface
on $k \in \mathbb{T}$ (by Theorem XIII.95 on p.301 in \cite{RS}), we
have $\omega_{l} \in H^s(\mathbb{T})$ for any $s \geq 0$, such that
$K_1 < \infty$. The second term in $\vec{\bf r}(\mbox{\boldmath
$\vec{\phi}$})$ is estimated as follows
\begin{eqnarray*}
\left\| \sum_{(n_1,n_2,n_3) \in \mathbb{Z}^3 \backslash \{(n,n,n)\}}
K_{n,n_1,n_2,n_3} \phi_{n_1} \bar{\phi}_{n_2} \phi_{n_3}
\right\|_{l^1(\mathbb{Z})} & \leq & \sum_{n \in \mathbb{Z}}
\sum_{(n_1,n_2,n_3) \in \mathbb{Z}^3 \backslash \{(n,n,n)\}}
|K_{n,n_1,n_2,n_3}| |\phi_{n_1}| |\phi_{n_2}| |\phi_{n_3}| \\ & \leq
& K_2 \| \vec{\mbox{\boldmath $\phi$}} \|^3_{l^1(\mathbb{Z})},
\end{eqnarray*}
where $K_2 = {\rm sup}_{(n_1,n_2,n_3) \in \mathbb{Z} \backslash
\{(n,n,n)\}} \sum_{n \in \mathbb{Z}} |K_{n,n_1,n_2,n_3}|$. Using the
exponential decay (\ref{decay-Wannier}), we obtain
$$
\sum_{n \in \mathbb{Z}} |\hat{u}_{l,n}(x)|  \leq C_{l} \sum_{n \in
\mathbb{Z}} e^{-\eta_{l} | x - 2 \pi n |} \leq A_l
$$
for some $A_l > 0$ uniformly in $x \in \mathbb{R}$ and
\begin{eqnarray*}
\sum_{n \in \mathbb{Z}} |K_{n,n_1,n_2,n_3}| \leq A_l
\int_{\mathbb{R}} |\hat{u}_{l,n_1}(x)| |\hat{u}_{l,n_2}(x)|
|\hat{u}_{l,n_3}(x)| dx \leq A_l \|
\hat{u}_{l,0}\|^2_{L^4(\mathbb{R})} \|
\hat{u}_{l,0}\|_{L^2(\mathbb{R})}
\end{eqnarray*}
uniformly in $(n_1,n_2,n_3) \in \mathbb{Z}^3$. By the Sobolev
embedding theorem and property (i) of Proposition
\ref{proposition-Wannier-limit}, $\| \hat{u}_{l,0}
\|_{L^4(\mathbb{R})} \leq C \| \hat{u}_{l,0} \|_{H^1(\mathbb{R})}
\leq C \| \hat{u}_{l,0} \|_{{\cal H}^1(\mathbb{R})}$ for some $C >
0$, such that $K_2 < \infty$. Therefore, the norm $\|
R(\mbox{\boldmath $\vec{\phi}$}) \|_{{\cal H}^1(\mathbb{R})}$ is
bounded from above by the norm $\| \vec{\mbox{\boldmath $\phi$}}
\|_{l^1(\mathbb{Z})}$.

To show that the constant $C_R$ is uniform for small $\mu > 0$, we
use Propositions \ref{proposition-bands} and
\ref{proposition-Wannier-limit}. By property (iii) of Proposition
\ref{proposition-bands}, $\hat{\omega}_{l,m} = {\rm O}(\mu^m)$ for
all $m \geq 2$, such that $K_1/\mu^2$ is uniformly bounded for
small $\mu$. By property (iii) of Proposition
\ref{proposition-Wannier-limit}, $K_{n,n_1,n_2,n_3} = {\rm
O}\left(\mu^{|n_1-n| + |n_2-n| + |n_3-n| + |n_2 - n_1| + |n_3 -
n_1| + |n_3 - n_2|}\right)$ for all $n_1,n_2,n_3 \in \mathbb{Z}^3
\backslash \{(n,n,n) \}$, such that $K_2/\mu$ is uniformly bounded
for small $\mu$. Thus, the first bound in (\ref{bound-map}) is
proved.

The second bound in (\ref{bound-map}) follows from the fact that
both ${\cal H}^1(\mathbb{R})$ and $l^1(\mathbb{Z})$ form Banach
algebras with respect to pointwise multiplication. As a result, if
$\vec{\mbox{\boldmath $\phi$}} \in l^1(\mathbb{Z})$ and
$\vec{\mbox{\boldmath $\phi$}}(T)$ is a solution of the DNLS
equation (\ref{DNLS}), then $\partial_T \vec{\mbox{\boldmath
$\phi$}} \in l^1(\mathbb{Z})$ and if $\varphi_0, \varphi_1 \in
{\cal H}^1(\mathbb{R})$ and $\vec{\mbox{\boldmath $\phi$}} \in
l^1(\mathbb{Z})$, then $N(\vec{\mbox{\boldmath $\phi$}},\psi)$
maps $\psi \in {\cal H}^1(\mathbb{R})$ to an element of ${\cal
H}^1(\mathbb{R})$.
\end{proof}

We can now prove that the initial-value problem for the
time-evolution equation (\ref{evolution-psi}) and the
initial-value problem for the DNLS equation (\ref{DNLS}) are
locally well-posed.

\begin{theorem}
\label{corollary-wellposedness} Let $\vec{\mbox{\boldmath
$\phi$}}(T) \in C^1(\mathbb{R},l^1(\mathbb{Z}))$ and $\psi_0 \in
{\cal H}^1(\mathbb{R})$. Then, there exists a $t_0 > 0$ and a
unique solution $\psi(t) \in C^1([0,t_0],{\cal H}^1(\mathbb{R}))$
of the time-evolution problem (\ref{evolution-psi}) with initial
data $\psi(x,0) = \psi_0(x)$.
\end{theorem}

\begin{proof}
Since $L$ is a self-adjoint operator, the operator $e^{-i t (L -
\hat{\omega}_{l,0})}$ forms a strongly continuous semi-group and
$$
\| e^{-i t (L - \hat{\omega}_{l,0})} \| \leq K_0,
$$
for some $K_0 > 0$ uniformly in $t \in \mathbb{R}_+$. Using the variation
of constant formula, we rewrite the time-evolution problem
(\ref{evolution-psi}) in the integral form
\begin{equation}
\label{integral-equation} \psi(t) = e^{-i t (L -
\hat{\omega}_{l,0})} \psi_0 + \int_0^t e^{-i (t-s) (L -
\hat{\omega}_{l,0})} \left( \mu R(\mbox{\boldmath
$\vec{\phi}$}(s)) + \mu \sigma N(\mbox{\boldmath
$\vec{\phi}$}(s),\psi(s)) \right) ds.
\end{equation}
By using bounds (\ref{bound-map}) on $R(\mbox{\boldmath
$\vec{\phi}$})$ and $N(\mbox{\boldmath $\vec{\phi}$},\psi)$ and
the contraction mapping principle for sufficiently small $t_0
> 0$, one can show with a standard analysis that there exists a unique solution
$\psi(t) \in C^1([0,t_0],{\cal H}^1(\mathbb{R}))$ of the integral
equation (\ref{integral-equation}).
\end{proof}

\begin{theorem}
\label{lemma-wellposedness} Let $\vec{\mbox{\boldmath $\phi$}}_0
\in l^1(\mathbb{Z})$. Then, there exist a $T_0 > 0$ and a unique
solution $\vec{\mbox{\boldmath $\phi$}}(T) \in C^1([0,
T_0],l^1(\mathbb{Z}))$ of the DNLS equation (\ref{DNLS}) with
initial data $\vec{\mbox{\boldmath $\phi$}}(0) =
\vec{\mbox{\boldmath $\phi$}}_0$.
\end{theorem}

\begin{proof}
By the variation of constant formula, we have
$$
\vec{\mbox{\boldmath $\phi$}}(T) = \vec{\mbox{\boldmath $\phi$}}_0
- i \int_0^T \left( \alpha \Delta \vec{\mbox{\boldmath $\phi$}}(s)
+ \sigma \beta \mbox{\boldmath $\Gamma$}(\vec{\mbox{\boldmath
$\phi$}}(s)) \right) ds,
$$
where $(\Delta \vec{\mbox{\boldmath $\phi$}})_n = \phi_{n+1} +
\phi_{n-1}$ and $(\mbox{\boldmath $\Gamma$}(\vec{\mbox{\boldmath
$\phi$}}))_n = |\phi_n|^2 \phi_n$. Since $l^1(\mathbb{Z})$ forms a
Banach algebra, the right-hand-side of the integral equation maps
an element of $l^1(\mathbb{Z})$ to an element of
$l^1(\mathbb{Z})$. Therefore, there exists a unique solution
$\vec{\mbox{\boldmath $\phi$}}(T) \in C^1([0,
T_0],l^1(\mathbb{Z}))$ of the integral equation for sufficiently
small $T_0 > 0$.
\end{proof}

We can now formulate the main theorem of our article.

\begin{theorem}
Fix $l \in \mathbb{N}$ and let $\vec{\mbox{\boldmath $\phi$}}(T)
\in C^1([0, T_0],l^1(\mathbb{Z}))$ be a solution of the DNLS
equation (\ref{DNLS}) with initial data $\vec{\mbox{\boldmath
$\phi$}}(0) = \vec{\mbox{\boldmath $\phi$}}_0$ satisfying the
bound
\begin{equation}
\label{bound-initial} \left\| \phi_0 - \mu^{1/2} \sum_{n \in
\mathbb{Z}} \phi_n(0) \hat{u}_{l,n}(x) \right\|_{{\cal
H}^1(\mathbb{R})} \leq C_0 \mu^{3/2}
\end{equation}
for some $C_0 > 0$. Then, for any $\mu \in (0,\mu_0)$ with
sufficiently small $\mu_0 > 0$, there exists a $\mu$-independent
constant $C > 0$ such that equation (\ref{GP}) has a solution
$\phi(t) \in C^1([0,T_0/\mu],{\cal H}^1(\mathbb{R}))$ satisfying
the bound
\begin{equation}
\label{bound-main} \forall t \in \left[ 0, T_0/\mu \right] : \quad
\left\| \phi(\cdot,t) - \mu^{1/2} \sum_{n \in \mathbb{Z}}
\phi_n(T) \hat{u}_{l,n} \right\|_{{\cal H}^1(\mathbb{R})} \leq C
\mu^{3/2}.
\end{equation} \label{theorem-main}
\end{theorem}

\begin{remark}
{\rm Since $\mu = \varepsilon e^{-a/\varepsilon}$, the finite
interval $[0,T_0/\mu]$ is exponentially large with respect to
parameter $\varepsilon$, similar to the bounds obtained in
\cite{BS}.}
\end{remark}

Theorem \ref{theorem-main} is proved in the following section.

\section{Bounds on the remainder terms}

We develop the proof of Theorem \ref{theorem-main} by using the
energy estimates for the time-evolution problem
(\ref{evolution-psi}) and Gronwall's inequality for a scalar
first-order differential equation. The GP equation (\ref{GP}) has
two conserved quantities
\begin{equation}
Q(\phi) = \int_{\mathbb{R}} |\phi|^2 dx, \quad E(\phi) =
\int_{\mathbb{R}} \left( |\phi_x|^2 + V(x) |\phi|^2 + \frac{1}{2}
\sigma |\phi|^4 \right) dx,
\end{equation}
which have the meaning of the charge and energy invariants, such
that $Q(\phi) = Q(\phi_0)$ and $E(\phi) = E(\phi_0)$ for any
solution $\phi(x,t)$ starting from the initial data $\phi_0(x)$.
We shall consider the quantity $E_Q(\psi) = \| \psi \|_{{\cal
H}^1(\mathbb{R})}$, which is not a constant in $t$ if $\psi(x,t)$
satisfies the time-evolution problem (\ref{evolution-psi}). The
time evolution of $E_Q(\psi)$ obeys the following estimate.

\begin{lemma}
Let $\mbox{\boldmath $\vec{\phi}$}(T) \in
C^1(\mathbb{R}_+,l^1(\mathbb{Z}))$ be any function and $\psi(t)
\in C^1([0,t_0],{\cal H}^1(\mathbb{R}))$ be a local solution of
the time-evolution problem (\ref{evolution-psi}) for some $t_0 >
0$. Then, for any $\mu \in [0,1]$ and every $M > 0$, there exist a
$\mu$-independent constant $C_E > 0$ such that
\begin{eqnarray}
\label{bound-1-2} \left| \frac{d}{d t}  E_Q(\psi) \right| \leq \mu
C_E \left( \| \mbox{\boldmath $\vec{\phi}$} \|_{l^1(\mathbb{Z})} +
E_Q(\psi) \right)
\end{eqnarray}
as long as $E_Q(\psi) \leq M$.
\label{time-estimate}
\end{lemma}

\begin{proof}
By direct differentiation, for any $\psi(t) \in C^1([0,t_0],{\cal
H}^1(\mathbb{R})$), we have
\begin{eqnarray*}
\frac{d}{d t} \| \psi \|^2_{{\cal H}^1(\mathbb{R})} & = & - i \mu
\int_{\mathbb{R}} \left( \bar{\psi}_x R_x(\mbox{\boldmath
$\vec{\phi}$}) - \psi_x \bar{R}_x(\mbox{\boldmath $\vec{\phi}$})
\right) dx - i \mu \sigma \int_{\mathbb{R}} \left( \bar{\psi}_x
N_x(\mbox{\boldmath $\vec{\phi}$},\psi) - \psi_x
\bar{N}_x(\mbox{\boldmath $\vec{\phi}$},\psi) \right) dx \\ &
\phantom{t} & - i \mu \int_{\mathbb{R}} \left( 1 + V(x) \right)
\left( \bar{\psi} R(\mbox{\boldmath $\vec{\phi}$}) - \psi
\bar{R}(\mbox{\boldmath $\vec{\phi}$}) \right) dx \\ & \phantom{t}
& - i \mu \sigma \int_{\mathbb{R}} \left( 1 + V(x) \right) \left(
\bar{\psi} N(\mbox{\boldmath $\vec{\phi}$},\psi) - \psi
\bar{N}(\mbox{\boldmath $\vec{\phi}$},\psi) \right) dx.
\end{eqnarray*}
Using the Cauchy--Schwartz inequalities and the bounds
(\ref{bound-map}) of Lemma \ref{lemma-reductions}, we obtain
\begin{eqnarray*}
\left| \frac{d}{d t} \| \psi \|^2_{{\cal H}^1(\mathbb{R})} \right|
& \leq & 6 \mu \| \psi \|_{{\cal H}^1(\mathbb{R})} \left( (C_R +
C_N) \| \mbox{\boldmath $\vec{\phi}$} \|_{l^1(\mathbb{Z})} + C_N
\| \psi \|_{{\cal H}^1(\mathbb{R})}\right),
\end{eqnarray*}
where $\sigma = \pm 1$ has been used. By canceling one power of
$\| \psi \|_{{\cal H}^1(\mathbb{R})}$, we arrive to the bound
(\ref{bound-1-2}).
\end{proof}

Theorem \ref{theorem-main} is then a direct consequence of the
following corollary.

\begin{corollary}
A local solution $\psi(t) \in C^1([0,T_0/\mu],{\cal
H}^1(\mathbb{R}))$ of the time-evolution problem
(\ref{evolution-psi}) for any $\psi(0) \in {\cal H}^1(\mathbb{R})$
and any $\mbox{\boldmath $\vec{\phi}$}(T) \in C^1([0,
T_0],l^1(\mathbb{Z}))$ satisfies the bound
\begin{equation}
\label{bound-time} \sup_{t \in [0,T_0/\mu]} \| \psi(t) \|_{{\cal
H}^1(\mathbb{R})} \leq \left( \| \psi(0) \|_{{\cal
H}^1(\mathbb{R})} + C_E T_0  \sup_{T \in [0,T_0]} \|
\mbox{\boldmath $\vec{\phi}$}(T) \|_{l^1 (\mathbb{Z})}  \right)
e^{C_E T_0}.
\end{equation}
\label{corollary-estimate}
\end{corollary}

\begin{proof}
By using the bound (\ref{bound-1-2}), we obtain
$$
E_Q(\psi(t)) \leq E_Q(\psi_0) + \mu C_E \int_0^t \left(   \|
\vec{\mbox{\boldmath $\phi$}}(\mu s) \|_{l^1(\mathbb{Z})} +
E_Q(\psi(s)) \right) ds.
$$
The bound (\ref{bound-time}) follows by Gronwall's inequality on
$t \in [0,T_0/\mu]$.
\end{proof}

\vspace{0.5cm}

{\bf Acknowledgement.} The work of D. Pelinovsky is supported by the
EPSRC and Humboldt Research Fellowships. The work of G. Schneider is
supported  by the Graduiertenkolleg 1294 ``Analysis, simulation and
design of nano-technological processes'' granted by the Deutsche
Forschungsgemeinschaft (DFG) and the Land Baden-W\"{u}rttemberg.

\end{document}